\title{Distributed Storage Systems based on\\ Equidistant Subspace Codes}
\author{Netanel Raviv \and Tuvi Etzion}
\date{\today}

\documentclass[12pt]{article}

\usepackage{filecontents}

\usepackage[margin=0.7in]{geometry}

\usepackage{amssymb}
\usepackage{amsmath}
\usepackage{amsthm}
\usepackage{units}

\newtheorem{theorem}{Theorem}
\newtheorem{definition}{Definition}
\newtheorem{observation}{Observation}

\newtheorem{lemma}{Lemma}
\newtheorem{corollary}{Corollary}

\newtheorem{construction}{Construction}
\newtheorem{example}{Example}

\newtheorem{algorithm}{Algorithm}

\numberwithin{subcase}{case}

\newtheorem{remark}{Remark}

\newcommand{\cl}[1]{\mathcal{#1}}

\newcommand{\bC}{\mathbb{C}}

\newcommand{\bF}{\mathbb{F}}

\newcommand{\bN}{\mathbb{N}}

\newcommand{\qbin}[3]{{#1 \brack #2}_{#3}}
\newcommand{\grsmn}[3]{\cl{G}_{#1}\left(#2,#3\right)}
\newcommand{\Span}[1]{{\left\langle {#1} \right\rangle}}

\DeclareMathOperator{\rank}{rank}

\makeatletter
\renewcommand*\env@matrix[1][*\c@MaxMatrixCols c]{%
  \hskip -\arraycolsep
  \let\@ifnextchar\new@ifnextchar
  \array{#1}}
\makeatother

\makeatletter
 \DeclareRobustCommand{\nsbinom}{\genfrac[]\z@{}}
 \makeatother

\newif\iffull
\fullfalse

\begin{document}
\maketitle

\begin{abstract}
Distributed storage systems based on equidistant constant dimension codes are presented. These equidistant codes are based on the Pl\"{u}cker embedding, which is essential in the repair and the reconstruction algorithms. These systems posses several useful properties such as high failure resilience, minimum bandwidth, low storage, simple algebraic repair and reconstruction algorithms, good locality, and compatibility with small fields.
\end{abstract}

\footnotetext[1]{This research was supported in part by the Israeli
Science Foundation (ISF), Jerusalem, Israel, under
Grant 10/12.

The work of Netanel Raviv is part of his Ph.D.
thesis performed at the Technion.

The authors are with the Department of Computer Science, Technion,
Haifa 32000, Israel.

e-mail: {etzion,netanel}@cs.technion.ac.il .}

\section{Introduction}\label{section:introduction}
Let $q$ be a prime power and let $\bF_q$ be the field with $q$ elements. In a distributed storage system (DSS) a file $x\in \bF_q^B$ is stored in $n$ storage \textit{nodes}, $\alpha$ information symbols in each. The DSS is required to be resilient to node failures; i.e., it should be possible to retrieve the data from a lost node by contacting~$d$ other active nodes and downloading $\beta$ information symbols from each one of them, an operation which is called \textit{repair}. In addition, a \textit{data collector} (DC) should be able to rebuild the stored file $x$ by contacting $k$ active nodes, an operation which is called \textit{reconstruction}. 
If the file is coded with an ordinary error correcting code $C$ prior to being stored in the system (usually by an MDS code \cite{ExactRepairMDS,SRCforDSS,LongMDSCodes,AccessVSBandwidth,
RepairOptimalErasure,InterferenceAlignment,ErrorResilienceinDSSNatalia,
OptimalLocallyRepairableCodesViaRankMetricCodes}), then $C$ is called the \textit{outer code}, and the DSS code is called the \textit{inner code}.

A repair process that results in a new node which contains the exact same information as in the failed node is called an \textit{exact repair} \cite{ExactRepairMDS,OptimalERcodesforDSS}. A repair process which is not an exact repair is called a \textit{functional repair}. Such a repair must maintain the system's ability of repair and reconstruction. The amount of data which is required for a repair is $d\beta$, and it is called the \textit{repair bandwidth} of the code. Codes which minimize the repair bandwidth, i.e., $d\beta=\alpha$, are called Minimum Bandwidth Regenerating (MBR) Codes \cite{NetworkCodingForDistributedStorage}. Codes which minimize $\alpha$, and thus have $\alpha = \frac{B}{k}$, are called Minimum Storage Regenerating (MSR) Codes \cite{NetworkCodingForDistributedStorage}. A Self-Repairing Code (SRC) \cite{SRCforDSSAProjective} is a code satisfying: (a) repairs are possible without having to download an amount of data equivalent to the reconstruction of the original file $x$; and (b) the number of nodes required for repair depends only on how many nodes are missing and not on their identity. 

In \cite{StorageCodes} a framework for a construction of a DSS code based on subspaces is given. This framework is slightly different from the classical one. In this framework every node $v_i$ is associated with a subspace $U_i$ of a vector space $U$ called the \textit{message space}. The dimension of $U$ is $B=|x|$, where $x\in\bF_q^B$ is the file to be stored. In the ``storage phase'' a node $v_i$ receives a vector $M_i\cdot x$, where $M_i$ is a full-rank matrix whose row span is $U_i$. A set of nodes is called a \textit{reconstruction set}\footnote{\cite{StorageCodes} uses the term \textit{recovery set}. We use a different term for consistency.} if their respective subspaces span the entire message space. The file $x$ is reconstructible from a reconstruction set $\{v_i\}_{i\in I},I\subseteq[n]$, where $[\ell]\triangleq\{1,\ldots,\ell\}$, by solving a linear nonsingular equation system based on $\left\{ M_i\cdot x\right\}_{i\in I}$ and $\left\{M_i\right\}_{i\in I}$. A set $\{v_i\}_{i\in T_j},T_j\subseteq[n]$ of nodes is called a \textit{repair set for a node $v_j$} if each subspace $U_i,i\in T_j$ contains a subspace $W_{i,j}\subseteq U_i$ such that the span of the set $\{W_{i,j}~\vert~i\in T_j\}$ contains $U_j$. The lost information $M_j\cdot x$ may be retrieved by manipulating the rows in a linear system based on $\left\{ M_{i,j}\cdot x\right\}_{i\in J}$ and $\left\{M_{i,j}\right\}$, where $M_{i,j}$ is a matrix whose row span is $W_{i,j}$. This framework yields an algebraic repair and reconstruction algorithms. We will use the equidistant subspace codes from \cite{Equidistant} as the subspaces in our DSS. We note that in this new framework the matrices $M_i$ have the role of the outer code in the classic framework.

Our codes achieve the SRC property, and nearly achieve the MSR and MBR properties. Regarding the MBR property, we show that $d\beta\le \alpha+1$, and hence the MBR property is achieved up to an additive constant of 1. Regarding the MSR property, we show that if the nodes participating in the reconstruction algorithm receive some information from the DC, then it is possible to reconstruct $x$ by communicating $|x|=B$ field elements, $\frac{b}{2}$ elements from each node if $b$ is even and either $\frac{b-1}{2}$ elements or $\frac{b+1}{2}$ elements if $b$ is odd. This property may be seen as a variant of the MSR property. Without this additional assumption it is possible to reconstruct $x$ by downloading $2B$ elements from $b$ nodes. The penalty of providing these advantages is not being able to repair (resp. reconstruct) from \textit{any} set of $d$ (resp. $k$) nodes, but rather some properly chosen ones. This drawback is also apparent in some existing DSS codes \cite{SRCforDSS,SRCforDSSAProjective}.

Our code stores a file $x \in \bF_q^B$, where $B={b \choose 2}$ for some $b\in\bN$, in $n$ nodes. The user may choose any $n$ such that $b\le n \le \frac{q^b-1}{q-1}$ in correspondence with the expected number of simultaneous node failures. Each node stores $b-1$ field elements. For the purpose of repair, the user may choose one of two possible algorithms. The first one requires that the \textit{newcomer node} (newcomer, in short) will contact either $b-1$ or $b$ active nodes and download a single field element from each one. This algorithm will minimize the \textit{repair bandwidth} as possible. The second algorithm requires downloading all data from as little as two nodes, depending on the code construction. In either of the algorithms it is not possible to contact \textit{any} set of nodes, but a proper set may be easily found, and it is promised to exist as long as the number of node failures does not exceed some reasonable bound.


The presented code has several useful properties. As mentioned earlier the user may choose between a local repair (Subsection \ref{section:LocalRepair}) and a minimum bandwidth repair (Subsection \ref{section:MinimumBandwidthRepair}). In addition, it is possible to reconstruct nodes that were not previously in the system (Corollary \ref{corollary:AutonomousRepair})
; that is, once a proper set of $b$ nodes is stored in the system by the user, the system may use repairs in order to generate additional storage nodes without any outside interference. It is also possible to repair in the presence of up to $O(\sqrt{B})$ simultaneous node failures, while imposing no restriction on the field size (Example \ref{example:Justesen}). Two additional useful properties are apparent. One is the ability to efficiently reuse the system to store a file $y\ne x$, without having to initialize all nodes (Subsection \ref{section:modification}). This property follows directly from the linear nature of our code. The second is the ability to simultaneously repair multiple node failures in parallel (Subsection \ref{section:ParallelRepair}).

A brief overview of the equidistant subspace codes from \cite{Equidistant} will be given in Section \ref{section:preliminaries}. The specific properties of our code strongly depend on an assignment of different vectors as identifiers to the storage nodes. The code will first be described with respect to a general assignment in Section~\ref{section:theAlgorithms}, and specific assignments, as well as their resulting properties, will be discussed in Section~\ref{section:Vectors}. Some proofs and further explanations in this version are omitted and will appear in the full version of this~paper.


\section{Preliminaries}\label{section:preliminaries}
The Grassmannian $\grsmn{q}{n}{k}$ is the set of all $k$-subspaces of $\bF_q^n$. The size of $\grsmn{q}{n}{k}$ is given by the Gaussian coefficient $\qbin{n}{k}{q}$ (see \cite[Chapter 24]{VanLintAndWilson}). A \textit{constant dimension code} (CDC) is a subset of $\grsmn{q}{n}{k}$ with respect to the \textit{subspace metric} $d_S(U,V)=\dim U +\dim V-2\dim(U\cap V)$. A CDC is called equidistant if the distance between every two distinct codewords is some fixed constant. An equidistant CDC is also called a \textit{$t$-intersecting} code since the dimension of the intersection of any two distinct codewords is some constant $t$. Our construction uses the $1$-intersecting equidistant subspace codes from \cite{Equidistant}, whose construction and properties are hereby described.

In what follows $e_i$ denotes the $i$th unit vector. For a set $S$ of vectors, $\left<S\right>$ denotes the linear span of $S$, and for a matrix $M$, $\left<M\right>$ denotes its row linear span.

\begin{definition}\label{definition:Plucker}
(The Pl\"{u}cker embedding, see \cite[Section 4]{Equidistant}, \cite[p. 165]{ProjectiveGeometryFromFoundationsToApplications}) Given $M\in \bF_q^{t\times b}$, identify the coordinates
of $\bF_q^{b \choose t}$ with all $t$-subsets of $[b]$, and define
$\varphi(M)$ as a vector of length ${b \choose t}$ in which
\[
\left( \varphi(M) \right) _{\left\{i_1,\ldots,i_t\right\}} \triangleq \det M\left( i_1 ,\ldots,i_t\right)
\]
where $M\left( i_1 ,\ldots,i_t\right)$ is the $t\times t$ sub-matrix
of $M$ formed from columns $i_1<\ldots<i_t$. For $v,u \in \bF_q^b$ we denote by $\varphi{v\choose u}$ the result of applying $\varphi$ on the $2\times b$ matrix $v \choose u$.
\end{definition}

\begin{definition}\label{definition:codeword}\cite[Subsection 3.1]{Equidistant}
For $V\in\grsmn{q}{b}{1}$, $v\in V\setminus\{0\}$, and the index $r(v)$ of the leftmost nonzero entry of $v$, let
\begin{eqnarray*}
P_V\triangleq \Span{\left\{\varphi{v\choose e_i}\right\}_{i\in[b]\setminus\{r(v)\}}}.
\end{eqnarray*}
\end{definition}
By the properties of the determinant function, any choice of a nonzero vector $v$ from the 1-subspace $ V$ results in the same subspace, and thus $P_V$ is well-defined. Lemma~\ref{lemma:OmitAnyNonzero} which follows shows that the choice of $r(v)$ as the leftmost nonzero entry of $v$ is arbitrary, and every other nonzero entry could equally be chosen.

\begin{theorem}\label{theorem:Equidistant}\cite[Theorem 14]{Equidistant}
The following code 
\begin{eqnarray*}
\bC \triangleq \left\{ P_V~|~V\in\grsmn{q}{b}{1}\right\},
\end{eqnarray*}
$\bC\subseteq \grsmn{q}{{b \choose 2}}{b-1}$ is an equidistant 1-intersecting code of size $\qbin{b}{1}{q}$; that is, any distinct $P_U,P_V\in\bC$ satisfy ${\dim(P_U\cap P_V)=1}$. In addition, for every distinct $P_U,P_V\in \bC$, $P_U\cap P_V = \Span{\varphi{u\choose v}}$, where $U=\Span{u}$ and $V=\Span{v}$.
\end{theorem}

The following lemma shows that the function $\varphi$ from Definition \ref{definition:Plucker} is a bilinear form when applied on two row matrices. This fact will be prominent in our constructions.

\begin{lemma}\cite[Lemma~4]{Equidistant} \label{lemma:linearPlucker}
If $v,u\in\bF_q^b$ are nonzero vectors, and $\gamma,\delta\in\bF_q$, then
$\varphi{v\choose \gamma u+\delta w}=\gamma \cdot \varphi{v\choose u}+\delta \cdot \varphi{v\choose w}$ and $\varphi{\gamma u+\delta w\choose v}=\gamma\cdot \varphi{u\choose v}+\delta \cdot \varphi{w\choose v}$.
\end{lemma}

Lemma \ref{lemma:computeMissing} and Lemma  \ref{lemma:OmitAnyNonzero} provide a convenient way of choosing a basis to any $P_V\in \bC$ (Theorem~\ref{theorem:Equidistant}); and both may easily be obtained from \cite[Lemma 3]{Equidistant}. For completeness we include a short proof.
\begin{lemma}\label{lemma:computeMissing}
If $v=(\gamma_1,\ldots,\gamma_b)\in\bF_q^b$ is a nonzero vector, then
\[
\sum_{j\in [b]}\gamma_j\cdot \varphi{v\choose e_j}=0.\]
\end{lemma}

\begin{proof}
By Lemma~\ref{lemma:linearPlucker} and by the properties of the determinant function we have
\begin{eqnarray*}
\sum_{j\in [b]}\gamma_j\cdot \varphi{v\choose e_j} & = & \varphi{v\choose \sum_{j\in [b]}\gamma_j e_j}
= \varphi{v\choose v}=0.
\end{eqnarray*}
\end{proof}

\begin{lemma}\label{lemma:OmitAnyNonzero}
If $v=\left(\gamma_1,\ldots,\gamma_b\right)\in \bF_q^b$, $\Span{v}\triangleq V$, and $\gamma_s\ne 0$ for some $s\in[b]$, then
\begin{eqnarray*}
P_V=\Span{\left\{\varphi{v \choose e_i}\right\}_{i\in[b]\setminus\{s\}}}.
\end{eqnarray*}
\end{lemma}

\begin{proof}
By Lemma~\ref{lemma:computeMissing}, 
\begin{eqnarray*}
\varphi{v \choose e_{r(v)}}&\in &\Span{\left\{\varphi{v\choose e_i}\right\}_{i\in[b]\setminus\{r(v)\}}}\\
\varphi{v \choose e_{s}}&\in &\Span{\left\{\varphi{v\choose e_i}\right\}_{i\in[b]\setminus\{s\}}},
\end{eqnarray*}
and hence,
\begin{eqnarray*}
P_V\triangleq\Span{\left\{\varphi{v\choose e_i}\right\}_{i\in[b]\setminus\{r(v)\}}}=\Span{\left\{\varphi{v\choose e_i}\right\}_{i\in[b]}}=\Span{\left\{\varphi{v\choose e_i}\right\}_{i\in[b]\setminus\{s\}}}
\end{eqnarray*}
\end{proof}

The following observation will be repeatedly used throughout our algorithms.

\begin{observation}\label{observation:vectorFromMatrix}
Let $A,B\in \bF_q^{b\times B}$ be two distinct row-equivalent matrices. If $r_1,\ldots,r_t$ is the series of row operations that transform $M_1$ to $M_2$, then for any $x\in\bF_q^B$ it is possible to compute $M_2x$ given $M_1x$ and $r_1,\ldots,r_t$.
\end{observation}

\begin{proof}
Let $E_1,\ldots,E_t$ be the invertible matrices corresponding to the row operations that transform $M_1$ to $M_2$; that is, $E_1 \cdot E_2 \cdot \ldots \cdot E_t \cdot M_1 = M_2$. The claim follows directly from the fact that for any $x\in\bF_q^B$, 
$E_1 \cdot E_2 \cdot \ldots \cdot E_t \cdot M_1x = M_2x$.
\end{proof}

\begin{remark}
The complexity analysis of Algorithms \ref{algorithm:MinimumBandwidthRepair} through \ref{algorithm:BRepair} in the sequel, relies mostly on the complexity of solving a system of linear equations over a finite field. This can be done either by a school book Gaussian elimination or by employing one of many faster algorithms (see \cite{LinearEquations} and references therein). However, to simplify the discussion we analyze our algorithms by using simple Gaussian elimination.
\end{remark}

\section{The 
Distributed Storage System}\label{section:theAlgorithms}
We are now in a position to describe the construction of the DSS. The feasibility of the described repair and reconstruction algorithms will depend on a certain assignment of vectors in $\bF_q^b$ to identify the storage nodes. Different assignments and their resulting parameters will be discussed separately in Section~\ref{section:Vectors}. With respect to a certain assignment of vectors to nodes, we will say that a set of nodes are \textit{linearly independent} if their assigned vectors are linearly independent.

\subsection{Storage}\label{section:storage}
Let $v_{1},\ldots,v_{n}$ be the available storage nodes. We identify each $v_i$ by a \textit{normalized} vector from $\bF_q^b$; that is, a vector whose leftmost nonzero entry $r(v_i)$ is 1. Let $M_{v_i}$ be the $(b-1)\times B$ matrix whose rows are the vectors \[\left\{\varphi{v_i\choose e_j}\right\}_{e_j \in [b]\setminus r(v_i)}.\]Following the terminology in \cite[Section III.A.]{StorageCodes}, each node $v_i$ is in fact associated with a subspace. In our system, this subspace is $P_{\Span{v_i}}\triangleq \left<M_{v_i}\right>$ (see Definition \ref{definition:codeword}).

Let $s$ be the source node, i.e. the node holding the file $x\in\bF_q^B$ to be stored. For the initial storage, $s$ sends $M_{v_i}\cdot x$ to $v_i$ for all $i=1,\ldots,n$. It is evident that $n\cdot (b-1)$ field elements are being sent. As for time complexity, computing the product $M_{v_i}\cdot x$ requires computing the matrix $M_{v_i}$. If the vector $v_i$ is given, each $\varphi{v_i \choose e_j}$ is computable from $v_i$ in $O(b\log b)$ time by using a proper sparse representation\footnote{e.g., a sparse representation of $x=\left(\gamma_1,\ldots,\gamma_B\right)$ is $\left\{\left(j,\gamma_j\right)\right\}_{j\vert\gamma_j\ne 0}$. This representation clearly requires ${O(w_H(x)\cdot \log B)=O(w_H(x)\cdot\log b)}$ space, where $w_H(x)$ is the Hamming weight of $x$.}. Hence, the matrix $M_{v_i}$ is computable in $O(b^2\cdot \log b)=O(B\log B)$. Using the same sparse representation, computing the product $M_{v_i}\cdot x$ takes an additional $O(B\log B)$ time for each~$v_i$. This stage requires $O(B\log B\cdot n)$ computation time and $O(B^{1/2}\cdot n)$ communication units.

\subsection{Minimum Bandwidth Repair}\label{section:MinimumBandwidthRepair}

In what follows we show that it is possible to repair a node failure by communicating a single field element from either $b-1$ or $b$ nodes. For functional repair no further computations are needed while for exact repair an additional $O(B^{2})$ algorithm should be applied by the newcomer.

\begin{algorithm}\label{algorithm:MinimumBandwidthRepair}
Let $v_j=\sum_{t=1}^b \gamma_t e_t$ be the failed node and let $u_{1},\ldots,u_{b'}$ be any set of active nodes such that $\Span{e_t}_{t\in[b]\setminus\{s\}}\subseteq \Span{u_{1},\ldots,u_{b'}}$ for some $s\in[b]$, where $\gamma_s\ne 0$ (obviously, $b-1\le b'\le b$). Each node $u_{\ell}$ computes 
\begin{eqnarray}\label{equation:MinimumBandwidthRepair}
\sum_{t=1}^b\gamma_t \varphi{u_{\ell}\choose e_t}\cdot x=\varphi{u_{\ell}\choose \sum_{t=1}^b\gamma_t e_t}\cdot x = \varphi{u_{\ell}\choose v_j}\cdot x,
\end{eqnarray}
and sends it to the newcomer.
\end{algorithm}

Notice that the elements \[\left\{\varphi{u_{\ell}\choose e_t}\cdot x\right\}_{t\in [b]\setminus r\left(u_{\ell}\right)}\] were sent to $u_{\ell}$ by $s$ in the initial stage (Subsection \ref{section:storage}). If needed, $\varphi{u_{\ell}\choose e_{r(u_{\ell})}}\cdot x$ may be computed using Lemma~\ref{lemma:computeMissing}. Hence, every node $u_\ell$ is capable of performing the computation in (\ref{equation:MinimumBandwidthRepair}).

\begin{lemma}\label{lemma:MinimumBandwidthrepair}
By using the information received from Algorithm \ref{algorithm:MinimumBandwidthRepair}, the newcomer may restore the information from the failed node $v_j$ by using $O(B^{2})$ field operations.
\end{lemma}

\begin{proof}
The newcomer may retrieve $M_{v_j}\cdot x$, the lost information of $v_j$, by using Lemma~\ref{lemma:linearPlucker}, Lemma~\ref{lemma:computeMissing}, and Lemma~ \ref{lemma:OmitAnyNonzero}. Since $\Span{e_j}_{j\in[b]\setminus\{s\}}\subseteq \Span{u_{1},\ldots,u_{b'}}$, it follows that the matrix 
\[
\begin{pmatrix}
\varphi{u_{1}\choose v_j}\\
\vdots \\
\varphi{u_{b'}\choose v_j}
\end{pmatrix}
\]
has a submatrix which is equivalent to
\[
A\triangleq -\begin{pmatrix}
\varphi{v_j\choose e_{i_1}}\\
\vdots \\
\varphi{v_j\choose e_{i_{b'}}}
\end{pmatrix},
\]
where $\{i_1,\ldots,i_{b'}\}=[b]\setminus\{s\}$. By Lemma~\ref{lemma:OmitAnyNonzero}, $\Span{A}=P_{\Span{v_j}}$ (see Definition \ref{definition:codeword}), and hence $A$ is row equivalent to $-M_{v_j}$. Therefore, the vector $M_{v_j}\cdot x$ may be extracted from the received information by using Observation \ref{observation:vectorFromMatrix}.

Assuming the identity of $v_j$ is known, this algorithm requires communicating either $b-1$ or $b$ field elements. For functional repair no further computations are required. For exact repair the newcomer needs to perform Gaussian-like process on a matrix of size $b'\times B$. By Lemma~\ref{lemma:linearPlucker}, this process requires the same $O(b^2)$ row operations preformed during a Gaussian elimination of a $b'\times b$ matrix. However, these row operations are being preformed on rows of length $B$, and hence this Gaussian elimination requires $O(b^2\cdot B)=O(B^{2})$ field operations.
\end{proof}

Notice that the only requirement on the nodes $u_{1},\ldots,u_{b'}$ participating in $v_j$'s repair process is that $\Span{e_j}_{j\in[b]\setminus\{s\}}\subseteq \Span{u_{1},\ldots,u_{b'}}$. It follows that if $u_{1},\ldots,u_{b'}$ are active nodes that form a basis to $\bF_q^b$ (i.e. $b'=b$), then it is possible to repair \textit{any} node $v_j$ by using Algorithm~\ref{algorithm:MinimumBandwidthRepair}. 

\begin{corollary}\label{corollary:AutonomousRepair}
Using Algorithm \ref{algorithm:MinimumBandwidthRepair}, it is possible to add a new node that was not initially in the DSS (see Section \ref{section:storage}).
\end{corollary}

\subsection{Local Repair}\label{section:LocalRepair}
It is often required that a failed node will be repairable from as few other active nodes as possible. It is clear that without replication of nodes, a minimum of two active nodes is necessary for such a repair. Clearly, such a repair can be done by contacting $k$ nodes from which the reconstruction is possible. In the following we present an alternative repairing approach that may achieve  this minimum. The possibility of achieving this minimum depends on the specific assignment of vectors to the nodes. This assignment will be discussed it detail in Section~\ref{section:Vectors}.

\begin{algorithm}\label{algorithm:LocalRepair}
Let $v_j$ be the failed node and let $\{u_{1},\ldots,u_{\ell}\}$ be a set of active linearly independent nodes such that $v_j \in \Span{u_{1},\ldots,u_{\ell}}$. For all $t\in[\ell]$, the newcomer $\nu$ downloads the entire vector $M_{u_{t}}\cdot x$ from $u_{t}$.
\end{algorithm}

\begin{lemma}\label{lemma:LocalRepair}
By using the information received from Algorithm \ref{algorithm:LocalRepair}, the newcomer $\nu$ may restore the information of the failed node $v_j$ in $O(\ell^2\cdot b)$ field operations.
\end{lemma}

\begin{proof}
Since $v_j\in\Span{u_{1},\ldots,u_{\ell}}$, it follows that $v_j=\sum_{t=1}^\ell \gamma_t u_{t}$ for some $\gamma_1,\ldots,\gamma_\ell \in \bF_q$. By the definition of the matrices $\{M_{u_{1}},\ldots,M_{u_{\ell}}\}$, $\nu$ downloads the set of elements
\begin{eqnarray*}
\left\{\varphi{u_{t}\choose e_i}\cdot x\right\}_{i\in[b]\setminus\{r(v_{i_t})\}}
\end{eqnarray*}
for all $t\in[\ell]$. The missing elements 
\[
\left\{\varphi{u_{t}\choose e_{r(u_{t})}}\cdot x\right\}_{t\in[\ell]}
\] are computed by Lemma~\ref{lemma:computeMissing} in $O(\ell\cdot b)$ field operations. The newcomer computes the coefficients $\gamma_1,\ldots,\gamma_\ell$, e.g. by performing Gaussian elimination on the matrix 
\begin{eqnarray*}
\begin{pmatrix}
u_{1}\\
\vdots\\
u_{\ell}\\
v_j
\end{pmatrix},
\end{eqnarray*}
a process requiring $O(\ell^2\cdot b)$ field operation. Having these coefficients the newcomer performs 
\begin{eqnarray}\label{equation:LocalRepair}
\sum_{t=1}^\ell \gamma_t \varphi{u_{t}\choose e_i}\cdot x&=&\varphi{\sum_{t=1}^\ell \gamma_t u_{t}\choose e_i}\cdot x\nonumber=\varphi{v_j\choose e_i}\cdot x.
\end{eqnarray}
for all $i\in[b]\setminus\{r(v_j)\}$ in $O(\ell\cdot b)$ operations, and reassembles the vector $M_j\cdot x$. Overall, Algorithm \ref{algorithm:LocalRepair} requires $O(\ell^2\cdot b)$ field operations and $\ell\cdot(b-1)$ communication units. 
\end{proof}

\begin{corollary}
Let $v_j$ be a failed node. If $\ell$ is the smallest integer such $v_j$ is in the linear span of $\ell$ other active nodes, then the locality of repairing $v_j$ is $\ell$.
\end{corollary}

\subsection{Parallel Repair}\label{section:ParallelRepair}
Consider the scenario of multiple simultaneous node failures. Obviously, under $t$ failures, if the conditions of Algorithm \ref{algorithm:LocalRepair} are satisfied, then it is possible to execute $t$ sequential instances of the repair algorithm. We show that this could be improved in a certain special case. This is a simple consequence of Lemma \ref{lemma:LocalRepair}.

\begin{lemma}\label{lemma:ParallelRepair}
If $\{v_{i_1},\ldots,v_{i_t}\}$ is a set of failed nodes and $\{v_{j_1},\ldots,v_{j_s}\}$ is a set of active linearly independent nodes, of the remaining nodes, such that 
\[
\{v_{i_1},\ldots,v_{i_t}\}\subseteq \Span{v_{j_1},\ldots,v_{j_s}},
\]
then it is possible to repair all failures by communicating $s\cdot (b-1)$ field elements.
\end{lemma}

\begin{proof}
Assume that a third party $\Psi$ is managing the repair process of all $t$ nodes simultaneously. $\Psi$~may download the entire content of all nodes $\{v_{j_1},\ldots,v_{j_s}\}$, and compute the set $\{\varphi{v_{i_m}\choose e_\ell}\cdot x\}_{\ell=1}^{b}$ for each $m\in[t]$ using Algorithm \ref{algorithm:LocalRepair}. 
\end{proof}

The complexity of Lemma~\ref{lemma:ParallelRepair} remains $t$ times the complexity of Algorithm \ref{algorithm:LocalRepair}. However, the amount of communication is the same as in a single instance of Algorithm \ref{algorithm:LocalRepair}. It is evident that this algorithm requires good locality. An assignment of vectors to nodes that achieves locality is discussed in Subection~\ref{section:MinimumLocalityAssignment}.

\subsection{Reconstruction}
This subsection presents two reconstruction algorithms for two different models of communication. In Algorithm~\ref{algorithm:2Brepair}, which follows, the DC accesses $b$ active nodes and downloads their data in its entirety for the reconstruction. The number of communicated field elements is $b(b-1)=2B$. Algorithm~\ref{algorithm:BRepair}, which follows, uses the additional assumption that the nodes participating in the reconstruction know the identities of one another (e.g., by broadcast, shared memory or by acknowledgement from the DC), and guarantees reconstruction  by communicating $B$ field elements. This is the minimum communication that guarantees a complete reconstruction of $x$.

\begin{algorithm}\label{algorithm:2Brepair}
Let $\{u_{1},\ldots,u_{b}\}$ be a set of active linearly independent nodes. For each $j\in [b]$, the DC downloads the vector $M_{u_{j}}\cdot x$ from $u_{j}$ and computes the missing element $\varphi{u_{j}\choose e_{r(u_{j})}}\cdot x$ from each node by using Lemma~\ref{lemma:computeMissing}. The DC assembles the vector $w\in\bF_q^{b^2}$ such that\footnote{The entries of the vector $w\in\bF_q^{b^2}$ are identified by the elements of $[b]^2$ according to the lexicographic order.} $w_{(i,j)}=\varphi{u_j \choose e_i}\cdot x$, and the $b^2\times B$ matrix $A$ whose rows are $\{\varphi{u_i\choose e_j}\}_{i,j\in[b]}$. The vector $x$ is then reconstructed by solving the linear system of equations $Ax=w$. 
\end{algorithm}

\begin{lemma}\label{lemma:Reconstruction1}
The matrix $A$ in Algorithm \ref{algorithm:2Brepair} has full rank. In particular, the DC may extract $x$ using $O(B^3)$ field operations.
\end{lemma}

\begin{proof}
By Lemmas~\ref{lemma:linearPlucker} and \ref{lemma:computeMissing}, for each $t\in[b]$ the submatrix 
\[
\begin{pmatrix}
\varphi{u_{1}\choose e_t}\\
\vdots\\
\varphi{u_{b}\choose e_t}\\
\end{pmatrix}
\]
is row equivalent to the matrix 
\[
\begin{pmatrix}
\varphi{e_1\choose e_t}\\
\vdots\\
\varphi{e_b\choose e_t}\\
\end{pmatrix}.
\]
The matrix $A$ is therefore row equivalent (up to redundant rows) to a matrix whose rows are $\{ \varphi{e_i\choose e_j} \}_{i\ne j}$, which may clearly be seen as equivalent to the identity matrix of size $B\times B$. Thus, the DC may use Observation \ref{observation:vectorFromMatrix} to recover $x$. Computing the rows of $A$ requires $O(b^2B)=O(B^2)$ operations. Solving a $b^2\times B$ linear system of equations requires additional $O(B^3)$ operations.
\end{proof}

Assuming that every node participating in the reconstruction algorithm knows the identity of all other participating nodes, it is possible to reduce the communication to merely $|x|=B$ field elements from $b-1$ nodes, $\frac{b}{2}$ elements from each node if $b$ is even and either $\frac{b-1}{2}$ elements or $\frac{b+1}{2}$ elements if $b$ is odd. As mentioned earlier, this is the minimum possible communication since no outer code is used. The following matrix, whose construction is deferred to Appendix A, will be used in Algorithm \ref{algorithm:BRepair}.

\begin{definition}\label{definition:GoodMatrix}
Let $N$ be a $b\times b$ matrix over $\bF_2$
such that 
\begin{enumerate}
\item[(1)] For all $i\in[b]$, $N_{i,b}=0$.
\label{definition:GoodMatrixZeroColumn}
\item[(2)] For all $i\in [b-1]$, $N_{b,j}=1$.
\label{definition:GoodMatrixOneRow}
\item[(3)] For all $i\in [b]$, $N_{i,i}=0$.
\label{definition:GoodMatrixZeroDiagonal}
\item[(4)] For all $i,j\in[b],i\ne j$, $N_{i,j}\ne N_{j,i}$.\label{definition:GoodMatrixAllElements}
\item[(5)] If $b$ is even then for all $i\in [b-1]$ the Hamming weight of the $i$th column is $\frac{b}{2}$.
\label{definition:GoodMatrixEvenb}
\item[(6)] If $b$ is odd then for all $i\in[b-1]$ the Hamming weight of the $i$th column is either $\frac{b-1}{2}$ or $\frac{b+1}{2}$ and the total Hamming weight of $N$ is $b\choose 2$.
\label{definition:GoodMatrixOddb}
\end{enumerate}
\end{definition}

\begin{example}\label{example:GoodMatrix}
The following matrices satisfy the requirements of Definition \ref{definition:GoodMatrix} for $b=6$ and $b=5$:
\[
\begin{pmatrix}
0&0&0&1&1&0\\
1&0&0&0&1&0\\
1&1&0&0&0&0\\
0&1&1&0&0&0\\
0&0&1&1&0&0\\
1&1&1&1&1&0\\
\end{pmatrix},
\begin{pmatrix}
0&0&1&1&0\\
1&0&0&1&0\\
0&1&0&0&0\\
0&0&1&0&0\\
1&1&1&1&0\\
\end{pmatrix}
\]
\end{example}

\begin{algorithm}\label{algorithm:BRepair}
Let $\{u_{1},\ldots,u_{b}\}$ be any set of linearly independent nodes and let $N$ be the matrix from definition \ref{definition:GoodMatrix}. For all $i\in[b-1]$, the DC downloads from node $u_i$ all elements $\varphi{u_i\choose u_j}\cdot x$ such that $N_{j,i}=1$. The DC assembles the vector $w\in\bF_q^{b\choose 2}$ such that\footnote{The entries of the vector $w\in\bF_q^{b\choose 2}$ are identified by all $2$-subsets of $[b]$ according to the lexicographic order.} $w_{\{i,j\}}=\varphi{u_i\choose u_j}\cdot x$, and a $B\times B$ matrix~$A$ whose rows are the vectors $\{\varphi{u_i\choose u_j}\}_{i\ne j}$. The vector $x$ is then reconstructed by solving the linear system of equations $Ax=w$.
\end{algorithm}

\begin{lemma}\label{lemma:Reconstruction2}
Using the information received from Algorithm \ref{algorithm:BRepair}, the DC may construct the vector $w$. In addition, the matrix $A$ from Algorithm \ref{algorithm:BRepair} has full rank, and hence $x$ is reconstructible by using $O(B^3)$ field operations.
\end{lemma}

\begin{proof}
By (4) of Definition \ref{definition:GoodMatrix} it is evident that for all $i,j\in[b]$, $i\ne j$, the DC receives the element $\varphi{u_i\choose u_j}\cdot x$ exactly once.
To prove that the reconstruction of $x$ is possible, we show that $A$ is row equivalent to a matrix whose rows are $\{\varphi{e_i\choose e_j}\}_{i\ne j}$. The latter may easily be seen as equivalent to the $B\times B$ identity matrix. For any $i\in[b]$, add the zero row $\varphi{u_i\choose u_i}$ to the proper submatrix to get
\[
\begin{pmatrix}
\varphi{u_{i}\choose u_1}\\
\vdots\\
\varphi{u_{i}\choose u_b}\\
\end{pmatrix}.
\]
By Lemma~\ref{lemma:linearPlucker} this matrix is row equivalent to
\[
\begin{pmatrix}
\varphi{u_i\choose e_1}\\
\vdots\\
\varphi{u_i\choose e_b}\\
\end{pmatrix}.
\]
By rearranging the rows of $A$ we may consider submatrices of the form
\[
\begin{pmatrix}
\varphi{u_1\choose e_i}\\
\vdots\\
\varphi{u_b\choose e_i}\\
\end{pmatrix}.
\]
for all $i\in[b]$. These submatrices are row equivalent by Lemma~\ref{lemma:linearPlucker} to 
\[
\begin{pmatrix}
\varphi{e_1\choose e_i}\\
\vdots\\
\varphi{e_b\choose e_i}\\
\end{pmatrix}.
\]
Hence, by addition of redundant rows, we get that $A$ is equivalent to the identity matrix. Thus, the reconstruction of $x$ is possible by Gaussian elimination, requiring $O(B^3)$ operations.
\end{proof}

\subsection{Modification} \label{section:modification}
A useful property of a DSS is being able to update a small fraction of $x$ without having to initialize the entire system. The linear nature of our code and the absence of an outer code allows these modifications to be done efficiently. In particular, the complexity of the process is a function of the Hamming distance $d_H(x,y)$, where $y$ is the modification of the vector $x$. In MDS based distributed storage systems a change of a single bit of $x$ usually requires changing a large portion of the data. Therefore, one more advantage of our system is revealed.

\begin{lemma}
If $x\in\bF_q^B$ is stored in the system, it is possible to update the system to contain $y\in\bF_q^B$ by communicating $\left(\log B+\log q\right)\cdot d_H(x,y)\cdot n$ bits.
\end{lemma}

\begin{proof}
Each node receives a list $\{\left(\delta_i,\ell_i\right)\}_{i=1}^{d_H(x,y)}$, where $\delta_i\in\bF_q$ and $\ell_i\in[B]$. The list indicates the values of the nonzero entries of the vector $y-x$. Each node $v$, holding the vector $M_v\cdot x$ (see Section~\ref{section:storage}) may assemble the matrix $M_v$ and compute:
\begin{eqnarray*}
M_v\cdot x + M_v\cdot(y-x)=M_v\cdot y.
\end{eqnarray*}
Communicating the list $\{\left(\delta_i,\ell_i\right)\}_{i=1}^{d_H(x,y)}$ to all the $n$ nodes clearly requires $\left(\log B+\log q\right)\cdot d_H(x,y)\cdot n$ bits.
\end{proof}

\section{Assignment of Vectors}\label{section:Vectors}

In Section \ref{section:theAlgorithms} we proved that the performance of the detailed algorithms strongly relies on the chosen vectors $v_1,\ldots,v_n$. Since both repair and reconstruction algorithms require linearly independent nodes, it follows that the assigned set of vectors should contain a basis to $\bF_q^b$ even after multiple failures. 

Choosing $n = \qbin{b}{1}{q}$ and assigning \textit{all} possible normalized vectors would suffice for repairing exponentially many failures. However, using $\qbin{b}{1}{q}=\Theta(q^b)$ storage nodes to store a file of size ${B=\Theta(b^2)}$ is unnecessary, as will be shown in the sequel. Furthermore, expecting exponentially many failures is nonrealistic.

In order to achieve reasonable failure resilience using a reasonable number of nodes, it suffices to consider the case $n=O(b)$. Subsection~\ref{section:MinimumBandwidthAssignment} discusses an assignment of vectors compatible with Algorithm \ref{algorithm:MinimumBandwidthRepair} presented in Subsection~\ref{section:MinimumBandwidthRepair}. An assignment compatible with Algorithms \ref{algorithm:LocalRepair} of Subsection~\ref{section:LocalRepair}~and~also for the algorithm of Subsection \ref{section:ParallelRepair} is presented in Subsection~\ref{section:MinimumLocalityAssignment}.

\begin{definition}\label{definition:tResilient}
For $t\in\bN$ a set $S\subseteq\bF_q^b$ is called a $t$-resilient spanning set if every $t$-subset $T$ of $S$ satisfies $\Span{S\setminus T}=\bF_q^b$.
\end{definition}

\begin{observation}
If $S$ is a $t$-resilient spanning set then by using $|S|$ storage nodes assigned with the vectors in $S$ (see Subsection \ref{section:storage}) then it is possible to repair and reconstruct in the presence of up to $t$ simultaneous node failures.
\end{observation}


\begin{example}
The following set is a $2$-resilient spanning set in $\bF_2^7$:
\[
\begin{array}{ccccccc}
1&0&0&0&0&0&0\\
0&1&0&0&0&0&0\\
0&0&1&0&0&0&0\\
0&0&0&1&0&0&0\\
0&0&0&0&1&0&0\\
0&0&0&0&0&1&0\\
0&0&0&0&0&0&1\\
1&1&1&1&1&1&1\\
1&1&1&1&0&0&0\\
1&1&0&0&1&1&0\\
1&0&1&0&1&0&1
\end{array}
\]
\end{example}


\subsection{Minimum Bandwidth Assignment} \label{section:MinimumBandwidthAssignment}

In what follows we present a construction of a set of vectors $\{v_1,\ldots,v_n\}$ compatible with Algorithm~\ref{algorithm:MinimumBandwidthRepair}, achieving $d\beta\le \alpha+1$.

\begin{lemma}\label{lemma:LinearCodeOne}
Let $b\in \bN$ and let $C$ be a \textit{linear} block code of length $c\cdot b$ for some constant $c>0$, dimension $b$, and minimum Hamming distance $\delta$ over $\bF_q$. If $M$ is a generator matrix of $C$ then the columns of $M$ are a $(\delta-1)$-resilient spanning set (see Definition~\ref{definition:tResilient}). 
\end{lemma}

The inverse of Lemma~\ref{lemma:LinearCodeOne} is also true, as stated in the next lemma.

\begin{lemma}
Let $S\subseteq \bF_q^b$ be an assignment of vectors to nodes in some DSS which is resilient to $t$ node failures by using the algorithms described in Section \ref{section:theAlgorithms}. If $G$ is the matrix whose columns are the elements of $S$ and $C\triangleq \{xG~\vert~x\in\bF_q^b\}$ then $C$ is a linear code of minimum Hamming distance~$t+1$.
\end{lemma}
\begin{example}\label{example:Justesen}
Let $C$ to a be binary Justesen code \cite{JustesenCodes} of length $O(b)$, dimension $b$, and minimum Hamming distance $\delta b$. We get that the corresponding code (see Section \ref{section:theAlgorithms}) uses $O(b)$ storage nodes while being able to recover from any $\delta b$ simultaneous node failures. In addition, the code uses the binary field. This choice admits the following parameters: $q=2$, $B={b \choose 2}$, $n=O(B^{1/2})$, $d=b=O(B^{1/2})$, $k=b=O(B^{1/2})$, $\alpha=b-1=O(B^{1/2})$, and $\beta=1$.
\end{example}

\subsection{Minimum Locality Assignment}\label{section:MinimumLocalityAssignment}
Algorithm \ref{algorithm:LocalRepair} in
Subsection \ref{section:LocalRepair} may possibly achieve the optimal locality. It is evident from Lemma~\ref{lemma:LocalRepair} that in order to get good locality, the set $\{u_1,\ldots,u_\ell\}$ from Algorithm \ref{algorithm:LocalRepair} is required to be small. However, this requirement conflicts with the requirements of Algorithms \ref{algorithm:MinimumBandwidthRepair}, \ref{algorithm:2Brepair}, and \ref{algorithm:BRepair}, since they all involve large linearly independent sets.

In this subsection we show that by choosing some basis of $\bF_q^b$, partitioning it to equally sized subsets and taking the linear span of each subset, some locality is achievable. The resulting failure resilience will grow with the field size. Thus, this technique will be particularly useful in large fields. 

\begin{definition}\label{definition:V}
Let $c$ be a positive integer such that $c$ divides $b$, and let $A\triangleq\{v_1,\ldots,v_b\}$ be a basis of~$\bF_q^b$. Partition $A$ into $\frac{b}{c}$ equally sized subsets $A_i\triangleq\{v_{ic+1},\ldots,v_{(i+1)c}\}$ for $i\in \{0,\ldots,\frac{b}{c}-1\}$. Let $V_i\subseteq \bF_q^b$ be a set of $\qbin{c}{1}{q}$ representatives for the 1-subspaces of $\left<A_i\right>$. Finally, let $V\triangleq \bigcup_{i=1}^{b/c}V_i$.
\end{definition}

\begin{lemma}
The set $V$ from Definition~\ref{definition:V} is a $\left(q^{c-1}-1\right)$-resilient spanning set (see Definition \ref{definition:tResilient}). Furthermore, assigning $V$ to nodes in a DSS allows repairing any node failure using at most $c$ active nodes. 
\end{lemma}

\begin{proof}
Since 
\begin{eqnarray*}
q^{c-1}-1=\qbin{c}{1}{q}-\qbin{c-1}{1}{q}-1<\qbin{c}{1}{q}-\qbin{c-1}{1}{q},
\end{eqnarray*}
it follows that after any set of at most $q^{c-1}-1$ node failures, the set of remaining active nodes in any $V_i$ is not contained in any $(c-1)$-subspace of $\left<A_i\right>$. Therefore, any $V_i$ still contains a basis for $\left<A_i\right>$. Since $\left<A_1\right>\oplus\cdots\oplus\left<A_{b/c}\right>=\bF_q^b$, it follows that $V$ is $(q^{c-1}-1)$-resilient spanning set. 

Let $v_j$ be a failed node and let $V_t$ be the set containing it. We have to prove that $v_j$ is repairable using at most~$c$ other nodes in the presence of at most~$q^{c-1}-1$ failures. We have shown that after $q^{c-1}-1$ failures, the remaining active nodes in any given $V_i$ contain a basis of $\left<A_i\right>$. Let $\{u_1,\ldots,u_c\}\subseteq\Span{A_t}$ be such a basis in $V_t$. It follows that $v_j \in \Span{u_1,\ldots,u_c}$, and hence $v_j$ is repairable by accessing at most~$c$ nodes by Lemma~\ref{lemma:LocalRepair}.
\end{proof}

This construction requires $\frac{b}{c}\cdot\qbin{c}{1}{q}$ nodes and allows locality of $c$ in the presence of up to $q^{c-1}-1$ failures. For simple comparison, the trivial replication code with $\frac{b}{c}\cdot\qbin{c}{1}{q}$ nodes allows locality of 1 in the presence of up to $\frac{1}{c}\cdot\qbin{c}{1}{q}-1$ failures. We note that 
\begin{eqnarray*}
\frac{q^{c-1}-1}{\frac{1}{c}\cdot\qbin{c}{1}{q}-1}\overset{q\to\infty}{\longrightarrow}c,
\end{eqnarray*}
and in particular for $c=2$,
\begin{eqnarray*}
\frac{q^{2-1}-1}{\frac{1}{2}\cdot\qbin{2}{1}{q}-1}=2.
\end{eqnarray*}
Therefore, this code outperforms the trivial one by approximately a factor of $c$ for large field size, while providing low locality. In particular, a minimal locality of 2 is achievable for any $q$.
\bibliographystyle{ieeetr}
\bibliography{../../Bibliography}

\section*{Appendix A}
Two constructions of a matrix satisfying the requirements of Definition \ref{definition:GoodMatrix} are given, Construction \ref{construction:evenb} for even $b$ and Construction \ref{construction:oddb} for odd~$b$. It is easily verified that these two constructions satisfy the requirement of Definition \ref{definition:GoodMatrix}.
\begin{construction}\label{construction:evenb}
Let $b$ be an even integer. Define $N\in\bF_2^{b\times b}$ as follows. For all $i\in [b]$ let $N_{i,b}=0$ and for all $i\in[b-1]$ let $N_{b,i}=1$. The remaining $(b-1)\times(b-1)$ upper left submatrix is defined as follows. The first row is the $b-1$ bit vector $0^{b/2}1^{b/2-1}$; that is, $\frac{b}{2}$ zeros followed by $\frac{b}{2}-1$ ones. The rest of the rows are all cyclic shifts of it (see Example \ref{example:GoodMatrix} for the case $b=6$).
\end{construction}
\begin{construction}\label{construction:oddb}
Let $b$ be an odd integer. Define $N\in \bF_2^{b\times b}$ as follows. For all $i\in [b]$ let ${N_{i,b}=0}$ and for all $i\in[b-1]$ let $N_{b,i}=1$. The remaining $(b-1)\times(b-1)$ upper left submatrix is defined as follows. The first row is the $b-1$ bit vector $0^{(b+1)/2}1^{(b-3)/2}$, that is, $\frac{b+1}{2}$ zeros followed by $\frac{b-3}{2}$ ones. The rest of the rows are all cyclic shifts of it. In addition, set the sub diagonal entries $(1,\frac{b-1}{2}+1),(2,\frac{b-1}{2}+2),\ldots,(\frac{b-1}{2},b-1)$ to 1 (see Example \ref{example:GoodMatrix} for the case $b=5$).
\end{construction}
\end{document}